\documentclass[10pt]{article}
\usepackage[utf8]{inputenc}

\usepackage{amssymb,amsthm,amsmath}
\usepackage{mathrsfs}
\usepackage{enumerate}
\usepackage[scr=boondox]{mathalpha}
\usepackage{graphicx,xcolor}
\usepackage{setspace}
\usepackage[hidelinks]{hyperref}
\newcommand{\dd}{\mathrm{d}}
\newcommand{\E}{\mathbb{E}}

\newcommand{\1}{\textbf{1}}
\newcommand{\R}{\mathbb{R}}

\newcommand{\cE}{\mathscr{E}}

\newcommand{\p}[1]{\mathbb{P}\left( #1 \right)}
\newcommand{\scal}[2]{\left\langle #1, #2 \right\rangle}

\DeclareMathOperator{\vol}{vol}

\newtheorem{theorem}{Theorem}
\newtheorem{lemma}{Lemma}
\newtheorem{corollary}{Corollary}

\theoremstyle{remark}
\newtheorem{remark}{Remark}

\theoremstyle{definition}

\title{\vspace{-3em}
Hermite--Fisher bounds and stability for min-entropy power inequalities
}

\author{Silouanos Brazitikos$^\dagger$, \ Martin Rapaport$^\S$, \
Tomasz Tkocz$^\S$\footnote{Email: ttkocz@math.cmu.edu. 
Research supported in part by 
Simons Foundation International Travel Support for Mathematicians grant SFI-MPS-TSM-00025936.}
}

\date{\begin{normalsize}
\emph{$^\dagger$Department of Mathematics and Applied Mathematics, University of Crete, 70013 Heraklion, Crete, Greece.} \\\vspace*{0.7em}
\emph{$^\S$Department of Mathematical Sciences, Carnegie Mellon University, Pittsburgh, PA 15213, USA}
\end{normalsize}}

\begin{document}

\maketitle

\begin{abstract}
We derive explicit lower bounds for relative Fisher information by combining a variational principle with suitably orthogonalized  Hermite-polynomial test functions. The resulting cumulant bounds are asymptotically sharp and yield lower bounds for Gaussian entropy deficits. We also establish quantitative versions of sharp min-entropy power inequalities in all dimensions. En route, we develop a stability result for Brzezinski's sharp bound for block sections of products of Euclidean balls, which may be of independent interest. 
\end{abstract}

\bigskip

\begin{footnotesize}

\end{footnotesize}

\bigskip

\section{Introduction}
The interaction between entropy, Fisher information and Gaussian approximation has a long history, going back to de Bruijn's identity and Stam's information-theoretic proof of the entropy-power inequality \cite{Stam}. Moment-based estimates for Fisher information were studied by Jarrett \cite{Jarr}, while Edgeworth expansions show that the first non-vanishing cumulants govern the leading asymptotic deviation from Gaussianity \cite{BCG-Fisher}. This naturally motivates the search for direct lower bounds on relative Fisher information in terms of a few low-order cumulants.

This paper has two main themes, both concerning sharp estimates for fundamental information-theoretic quantities.

First, we derive a simple general lower bound for relative Fisher information from an elementary variational principle. The bound is particularly well suited to weighted sums of independent random variables, since it depends only on a few low-order moments of the summands. 
The study of important probabilistic characteristics of such sums has attracted considerable attention, in a variety of contexts, ranging from analytic, or geometric in nature (like concentration), \cite{BCG-book, BLM, Led} to those more aligned with the theme of this work, see, e.g. \cite{ABBN1, ABBN2, EG, ENT-GM}.
We apply the method to centred exponential random variables and to coordinate marginals of the uniform distribution on a sphere, obtaining explicit asymptotically sharp estimates.  Integrating these estimates along the Ornstein--Uhlenbeck flow also yields explicit lower bounds for the Gaussian entropy deficit.

%The key feature of our approach is that, by choosing Hermite polynomials as test functions in the variational principle, the resulting bounds can be expressed explicitly in terms of low-order cumulants. Since cumulants are additive under convolution, this gives effective estimates for weighted sums of independent random variables. In particular, the method detects the cancellation of the third cumulant caused by the coefficients and shows how, in that case, the fourth cumulant produces the leading contribution to the relative Fisher information.

Second, we turn to entropy-power inequalities beyond the Shannon setting. Shannon's EPI established a fundamental superadditivity principle for entropy under convolution and became a model for a much broader programme. R\'enyi's one-parameter family of $p$-entropies naturally led to the search for corresponding inequalities at other entropy orders, and Bobkov and Chistyakov obtained such a family for all $p>1$, including, at the endpoint $p=\infty$, a min-entropy power inequality, or equivalently an anti-concentration estimate \cite{BCh}. 
The optimal constants at this endpoint were subsequently identified in dimension one by Bobkov and Chistyakov and in arbitrary dimension by Madiman, Melbourne
and Xu, while Melbourne and Roberto initiated the accompanying study of
stability and equality in the scalar case \cite{BCh-max, MMX, MR}. 

Once the sharp constants are known, the natural next step is to understand rigidity: which configurations attain equality, and how does the deficit control deviation from extremality? We pursue this question by establishing global quantitative deficit estimates for the sharp min-entropy power inequality in every dimension, characterising the equality cases in dimensions $1$ and $2$, and quantifying the
strict non-attainment for finite sums in dimensions $d\geq3$.

The paper is organised as follows. In
Section~\ref{sec:preliminaries}, we recall the relevant notions of
R\'enyi entropy and Fisher information, together with de Bruijn's
identity. In Subsection~\ref{sec:variational}, we establish the
variational principle and its Hermite-polynomial consequences. In
Subsection~\ref{sec:weighted-sums}, we apply these estimates to weighted
sums, with applications to centred exponential random variables and to
coordinate marginals of the uniform distribution on the sphere. Finally,
Section~\ref{sec:min-entropy} is devoted to stability and equality in
min-entropy power inequalities.

\section{Entropy and Fisher information}
\label{sec:preliminaries}

We begin by recalling several definitions and concepts fundamental in information theory. Given $p \in [0, +\infty]$, for a random vector $X$ in $\R^d$ with density $f$ (with respect to Lebesgue measure), we define its $p$-R\'enyi entropy as
\[ 
h_p(X) = h_p(f) = \frac{1}{1-p}\log\left(\int_{\R^d} f^p\right)
 \]
(provided the integral exists) and the $p$-R\'enyi entropy power
\[
N_p(X) = \exp\left(\frac{2}{d}h_p(X)\right)
 \]
with the values at $p = 0, 1, +\infty$ understood as the expressions (formally arising by taking the limits)
\begin{align*}
h_0(X) &= \log|\text{supp}(f)|, \\
h(X) &= h_1(X) = -\int_{\R^d} f \log f, \\
h_\infty(X) &= -\log\|f\|_\infty,
\end{align*}
the max-, Shannon- and min-entropy, respectively. We also abbreviate $N \equiv N_1$. We refer for instance to \cite{CT,Renyi1961} for further background. Nota bene, $h_p(X) = -\log\|f(X)\|_{L_{p-1}(X)}$, hence $p \mapsto h_p$ is nonincreasing, explaining the terminology of min/max-entropy. 

For densities $f$ with $\sqrt{f} \in H^1(\R^d)$, the Hilbert space of functions with square-integrable weak (distributional) derivatives, we define the Fisher information as
\[ 
I(X) = 4\int_{\R^d} |\nabla \sqrt{f}|^2 = \int_{\R^d} \frac{|\nabla f|^2}{f}.
 \]
It is fundamental and classical that among all densities with $\int |x|^2f(x) = d$, the standard Gaussian density $\phi(x) = (2\pi)^{-d/2}\exp(-\frac12|x|^2)$, $x 
\in \R^d$ has the least information. Thus, as a natural measure of closeness to Gaussianity, we define the relative Fisher information,
\[ 
J(X) = 4\int_{\R^d} \left|\left(\nabla + \frac{x}{2}\right) \sqrt{f(x)}\right|^2  \dd x= I(f) - I(\phi) = I(X) - d,
 \]
under the normalisation $\E X = 0$, $\E |X|^2 = d$. Here and throughout, $|\cdot|$ denotes the standard Euclidean norm.

We set $I(X) = +\infty$ and $J(X) = +\infty$, whenever $X$ does not possess a density $f$ with $\sqrt{f} \in H^1$.

A vital link between entropy and information was discovered by de Bruijn (in the literature presented and used by Stam in his seminal proof of the EPI in \cite{Stam}). The identity states that for the heat-flow evolution of $X$,
\[ 
X_t = X + \sqrt{t}Z, \qquad t > 0,
 \]
where $Z$ is a standard Gaussian random vector (i.e. with density $\phi$), independent of $X$, as long as $X$ has finite variance, the entropy of $X_t$ is finite, differentiable on $(0, +\infty)$, and
\[ 
\frac{\dd }{\dd t}h(X_t) = \frac12I(X_t).
 \]
Equivalently, for the Ornstein-Uhlenbeck evolution,
\[ 
\tilde X_t = e^{-t}X + \sqrt{1-e^{-2t}}Z, \qquad t > 0,
 \]
we have
\[ 
\frac{\dd }{\dd t}h(\tilde X_t) = J(\tilde X_t),
 \]
provided that $\E|X|^2 = d$. We refer to classical textbooks, e.g. \cite{CT} for background.

\subsection{A variational principle and Hermite-Fisher bound}
\label{sec:variational}

A starting point for us is a standard variational estimate for the relative Fisher information. It is the Gaussian-relative analogue of the integration-by-parts and Cauchy-Schwarz argument used for ordinary Fisher information, for instance see \cite{BCG-Fisher}, Corollary 2.3 therein and its proof (derived and used therein in a different context of bounds involving the characteristic function).

\begin{theorem}\label{thm:Fisher-main}
Let $X$ be a random vector in $\R^d$ with density $f$ on $\R^d$ satisfying $\sqrt{f} \in H^1$. Suppose $\E X = 0$, $\E|X|^2 = d$. For every $C^1$ compactly supported vector field $\Psi\colon \R^d \to \R^d$ with $\E|\Psi(X)|^2  > 0$, we have
\begin{equation}\label{eq:Fisher-variational}
J(X) \geq \frac{\Big(\E[\scal{X}{\Psi(X)} - \mathrm{div}\Psi(X)]\Big)^2}{\E[|\Psi(X)|^2]}.
\end{equation}
\end{theorem}
\begin{proof}
Using the Cauchy-Schwarz inequality, we get
\begin{align*}
\sqrt{J(X)}\cdot \sqrt{\E[|\Psi(X)|^2]} &= \sqrt{\int_{\R^d} \left|\left(2\nabla + x\right) \sqrt{f(x)}\right|^2  \dd x}\cdot\sqrt{\int_{\R^d}|\Psi(x)\sqrt{f(x)}|^2 \dd x} \\
&\geq \left|\int_{\R^d} \scal{\left(2\nabla + x\right) \sqrt{f(x)}}{\Psi(x)\sqrt{f(x)}} \dd x \right|\\
&= \left|\int_{\R^d} \left( \scal{\nabla f(x)}{\Psi(x)} + \scal{x}{\Psi(x)}f(x)\right) \dd x\right|.
\end{align*}
Since $\Psi$ is compactly supported, integration by parts gives  $-\E[\mathrm{div}\Psi(X)]$ for the first term. The second term is $\E[\scal{X}{\Psi(X)}]$. Squaring completes the proof.
\end{proof}

When specialised to Hermite polynomials, this gives a strong bound, particularly handy for sums of independent random variables. 
We illustrate the concept in dimension one. Recall that the (probabilists') Hermite polynomials $(H_k)_{k=0}^\infty$ are defined by
\[ 
\frac{\dd^k}{\dd x^k}e^{-x^2/2} = (-1)^kH_k(x)e^{-x^2/2}.
 \]
That is,
\begin{align*}
H_0 \equiv 1, \quad H_1(x) = x, &\quad H_2(x) = x^2-1, \\
H_3(x) = x^3-3x, &\quad H_4(x) = x^4-6x^2+3, \quad\dots
 \end{align*}
We note the standard but very useful recurrence relation
\[ 
xH_k(x) - H_{k}'(x) = H_{k+1}(x), \qquad k \geq 0.
 \]
 
%Applying \eqref{eq:Fisher-variational} to $\Psi = H_k$ (approximated by a truncation with a smooth compactly supported cutoff function) thus yields the following bound.
For $a,b\in\mathbb{R}$, apply \eqref{eq:Fisher-variational}  to
\[
    \Psi_{a,b}(x)=H_k(x)-a-bx,
\]
using a smooth compactly supported truncation. Since
$\mathbb{E}X=0$ and $\mathbb{E}X^2=1$, the Hermite recurrence gives
\[
\begin{aligned}
\mathbb{E}\bigl[X\Psi_{a,b}(X)-\Psi_{a,b}'(X)\bigr]
&=\mathbb{E}H_{k+1}(X)
  -a\mathbb{E}X-b\mathbb{E}(X^2-1) \\
&=\mathbb{E}H_{k+1}(X).
\end{aligned}
\]
Thus the numerator in \eqref{eq:Fisher-variational} is independent of the parameters $a$ and $b$, and we may choose
them to minimize the denominator.

\begin{corollary}\label{cor:Hk}
Let $X$ be a mean $0$, variance $1$ random variable. Let $k\geq 2$
be an integer and suppose that $\E|X|^{2k}<\infty$. Define
\[
\widetilde H_k(x)
=
H_k(x)-\E H_k(X)-\E[XH_k(X)]x.
\]
If $\E[\widetilde H_k(X)^2]>0$, then
\begin{equation}\label{eq:Hermite-Fisher}
J(X)
\geq
\frac{\bigl(\E H_{k+1}(X)\bigr)^2}
     {\E[\widetilde H_k(X)^2]}
=
\frac{\bigl(\E H_{k+1}(X)\bigr)^2}
{
 \E[H_k(X)^2]
 -\bigl(\E H_k(X)\bigr)^2
 -\bigl(\E[XH_k(X)]\bigr)^2
}.
\end{equation}
Otherwise, $X$ is supported on the finite zero set of the nonzero
polynomial $\widetilde H_k$, hence $J(X)=+\infty$.

In particular, whenever the indicated denominators are nonzero, the
choices $k=2$ and $k=3$ give
\begin{equation}\label{eq:Hermite-Fisher-H2}
J(X)
\geq
\frac{\bigl(\E X^3\bigr)^2}
     {\E[(X^2-1-\kappa_3(X)X)^2]}
=
\frac{\kappa_3(X)^2}
     {2+\kappa_4(X)-\kappa_3(X)^2},
\end{equation}
and
\begin{equation}\label{eq:Hermite-Fisher-H3}
\begin{aligned}
J(X)
&\geq
\frac{\bigl(\E X^4-3\bigr)^2}
     {\E[(X^3-3X-\kappa_3(X)-\kappa_4(X)X)^2]}
\\
&=
\frac{\kappa_4(X)^2}
     {6+9\kappa_3(X)^2+9\kappa_4(X)+\kappa_6(X)
      -\kappa_4(X)^2}.
\end{aligned}
\end{equation}
\end{corollary}
\begin{proof}
Since $X$ has mean $0$ and variance $1$, the functions $1$ and $x\mapsto x$ are orthonormal in $L_2(\mu_X)$, $\mu_X$ denoting the law of $X$.
Consequently, the orthogonal projection of $H_k(x)$ onto $\operatorname{span}\{1,x\}$ is $\E H_k(X) + \E[X H_k(X)]x$, thus
\[
\inf_{a,b\in\R}\E[(H_k(X)-a-bX)^2]
=
\E[\widetilde H_k(X)^2],
\]
and the minimum is attained at
\[
a=\E H_k(X),
\qquad
b=\E[XH_k(X)].
\]
As observed above, subtracting $a+bX$ does not change the
numerator in \eqref{eq:Fisher-variational}, inequality
\eqref{eq:Hermite-Fisher} follows.
\end{proof}

Here and throughout, $\kappa_r(X)$ denotes the $r$th cumulant of $X$,
$r\geq 0$. When the moment generating function $\E e^{tX}$ is
well-defined in a neighbourhood of $0$, the cumulants are the
coefficients in the Taylor expansion of the log-moment generating
function,
\[
\Lambda_X(t)
=
\log(\E e^{tX})
=
\sum_{r=0}^\infty \kappa_r(X)\frac{t^r}{r!}.
\]
In general, under only a finite moment assumption
$\E|X|^\ell<\infty$, they can be formally defined through the
characteristic function,
\[
\kappa_r(X)
=
i^{-r}\frac{\dd^r}{\dd t^r}\Big|_{t=0}\log\E e^{itX},
\qquad
1\leq r\leq\ell.
\]
We recall that using the standard relations between the central
moments and cumulants (obtained for instance from the formal identity
$\E[Xe^{tX}]=\Lambda_X'(t)\E e^{tX}$), for the mean $0$, variance $1$
random variable $X$, we have
\begin{align*}
\E X^2 &= \kappa_2(X)=1,
&\qquad \E X^3 &= \kappa_3(X),\\
\E X^4 &= 3+\kappa_4(X),
&\qquad \E X^6
&= \kappa_6(X)+15\kappa_4(X)
   +10\kappa_3(X)^2+15.
\end{align*}
For $k=2$ and $k=3$, the coefficients of the corresponding affine
projections satisfy
\[
\begin{aligned}
\E H_2(X)&=0,
&\qquad \E[XH_2(X)]&=\kappa_3(X),
\\
\E H_3(X)&=\kappa_3(X),
&\qquad \E[XH_3(X)]&=\kappa_4(X).
\end{aligned}
\]
Together with the preceding moment identities, these relations yield
\eqref{eq:Hermite-Fisher-H2} and
\eqref{eq:Hermite-Fisher-H3}.

Polynomial projection methods for Fisher information go back at least to
Jarrett's work \cite{Jarr} on Fisher information with known moments (see Section 2 and Theorem 3.1 therein). In the CLT
setting, Hermite polynomials enter naturally through the Edgeworth
expansions of Bobkov--Chistyakov--Götze from \cite{BCG-Fisher}.  The present argument combines these two ideas in a particularly simple
way: we project the relative (to Gaussian) score function
$\frac{\nabla f(X)}{f(X)}+X$ onto a single orthogonalized Hermite
direction $\widetilde H_k$.

\subsection{Bounds for weighted sums}
\label{sec:weighted-sums}

Since the cumulants are additive under sums of independent random variables, Corollary \ref{cor:Hk} is especially well-suited for such sums. Our main result below follows at once from \eqref{eq:Hermite-Fisher-H3}.

\begin{theorem}\label{thm:sums}
Let $X_1, X_2, \dots,X_n$ be i.i.d. mean $0$, variance $1$ random variables with $\mathbb E|X_1|^6<\infty$. Let $a_1, a_2, \dots,a_n$ be real numbers such that $\sum_{j=1}^n a_j^2 = 1$. We have,
\begin{equation}\label{eq:J-sums}
J\left(\sum_{j=1}^n a_jX_j\right)
\geq
\frac{\kappa_4^2s_4^2}
{
6+9\kappa_4s_4+9\kappa_3^2s_3^2+\kappa_6s_6
-\kappa_4^2s_4^2
}.
\end{equation}

where $\kappa_r = \kappa_r(X_1)$ are the cumulants of $X_1$ and
$s_r = \sum _{j=1}^{n}a_j^r.$
In particular, 
\begin{align}\label{eq:J-asympt}
J\left(\sum_{j=1}^n a_jX_j\right)
&\geq
\frac{\kappa_4^2}{6n^2}
\left(
1+\frac{9|\kappa_4|+9\kappa_3^2}{6n}
+\frac{|\kappa_6|}{6n^{3/2}}
\right)^{-1}
\\
&=
\frac{\kappa_4^2}{6n^2}+O_{X_1}(n^{-3}).
\notag
\end{align}
\end{theorem}
\begin{proof}
Let $X=\sum_j a_jX_j$. Plainly,
$\kappa_r(X)=\kappa_rs_r$, hence
\eqref{eq:Hermite-Fisher-H3} becomes \eqref{eq:J-sums}.
To obtain \eqref{eq:J-asympt}, we discard the nonpositive term
$-\kappa_4^2s_4^2$ in the denominator and use
\[
s_3^2\leq s_4s_2=s_4,
\qquad
s_6\leq s_4^{3/2}.
\]
This gives
\[
J(X)
\geq
\frac{\kappa_4^2s_4^2}
{
6+(9|\kappa_4|+9\kappa_3^2)s_4
+|\kappa_6|s_4^{3/2}
}.
\]
Another application of the Cauchy-Schwarz inequality gives $1 = s_2^2 \leq ns_4$, so $s_4 \geq \frac1n$. It remains to note that whenever $c > 0$, $b_1, b_2, \dots \geq 0$ and $0 < \alpha_1, \alpha_2, \dots < 2$, the function
\[ 
(0,\infty) \ni u \mapsto \frac{u^2}{c+\sum b_ju^{\alpha_j}}
 \]
is strictly increasing. Indeed,
\[ 
\frac{\dd }{\dd u}\frac{u^2}{c + \sum b_ju^{\alpha_j}} = \frac{2cu + \sum (2-\alpha_j)b_ju^{\alpha_j+1}}{(c+\sum b_ju^{\alpha_j})^2} > 0.
 \]
When all weights $a_j$ equal and $\kappa_3=0$, the main result from \cite{BCG-Fisher} (Corollary 1.2) asserts in particular that $J(X) = \frac{\kappa_4^2}{6}n^{-2} + o(n^{-2})$, hence the claimed sharpness.
\end{proof}

\begin{remark}\label{rem:sharpness}
The order $1/n^2$ of the leading term as well as the constant $\kappa_4^2/6$ in \eqref{eq:J-asympt} are in general best possible. Specifically, if the hypotheses of the Fisher-information Edgeworth expansion
of Bobkov, Chistyakov and Götze are satisfied (Corollary 1.2 in \cite{BCG-Fisher}) -- finite $6$th moment and eventual finite Fisher information along the CLT sums--, then for equal weights and $\kappa_3 = 0$,
\[
J\left(\frac1{\sqrt n}\sum_{j=1}^nX_j\right) =
\frac{\kappa_4^2}{6n^2}+o(n^{-2}).
\]
With sufficiently many moments, one can replace the $o(n^{-2})$ term with explicit remainders of higher order.
Consequently, when $\kappa_4 \ne 0$, both the order $n^{-2}$ and the
leading constant $\kappa_4^2/6$ are optimal. 
\end{remark}

\begin{remark}\label{rem:kappa3=0}
When $\kappa_3 \neq 0$, bound \eqref{eq:J-asympt} may not be tight, depending on the behaviour of $s_3$. An improvement can be obtained by repeating verbatim the approach with the bound \eqref{eq:Hermite-Fisher-H2} in lieu of \eqref{eq:Hermite-Fisher-H3}, leading to
\begin{align*}
J\left(\sum a_jX_j\right) &\geq
\frac{\kappa_3^2s_3^2}
{2+\kappa_4s_4-\kappa_3^2s_3^2}.
\end{align*}
When the $a_j$ are all equal, this leads to
\[
J\left(\sum a_jX_j \right) \geq
\frac{\kappa_3^2}{2n}
\big(1+(\kappa_4-\kappa_3^2)/(2n)\big)^{-1}
=
\frac{\kappa_3^2}{2n}+O(n^{-2})
\]
which has the sharp leading term $\kappa_3^2/(2n)$, under the hypotheses
of the corresponding Edgeworth expansion. For arbitrary signed weights,
however, cancellation may make $s_3$ arbitrarily small.
\end{remark}

The main point, however, is that the relevant cumulants are those of the sum $X= \sum a_jX_j$ rather than of a single summand $X_1$, viz. the above bound may be vacuous when $s_3 = 0$, even though $\kappa_3 \neq 0$. In other words, cancellation of the third cumulant removes the usual $1/n$ Fisher-information contribution, and the fourth cumulant typically produces the leading $1/n^2$-term, as provided by the $H_3$-bound \eqref{eq:Hermite-Fisher-H3}.

This phenomenon is illustrated by the example of centred exponential random variables, which we present in the next corollary. Nota bene, this example served as the main motivation for this work: the weighted sums of exponential random variables are naturally connected to the volume of sections of the regular simplex (see \cite{AG, BP1, BP2, BTT, MRTT, MTTT, Tang, Webb}).

\begin{corollary}\label{cor:exps}
Let $\cE_1, \cE_2, \dots, \cE_n$ be i.i.d. standard exponential random variables, i.e. with density $e^{-x}\1_{(0,+\infty)}$ on $\R$. Let $X_j = \cE_j - 1$ be mean $0$, variance $1$ shifted exponential random variables and $a_1, a_2, \dots,a_n$ be real numbers such that $\sum_{j=1}^{n} a_j^2 = 1$. For
\[ 
X = \sum_{j=1}^{n} a_jX_j,
 \]
one has
\begin{align*}
J\left(X\right)
&\geq
\frac{36s_4^2}
{6+54s_4+36s_3^2+120s_6-36s_4^2}
\\
&\geq
\frac{6}{n^2}
\left(
1+\frac{15}{n}+\frac{20}{n^{3/2}}-\frac{6}{n^2}
\right)^{-1}
=
\frac{6}{n^2}-\frac{90}{n^3}+O(n^{-7/2}).
\end{align*}
\end{corollary}

The leading term $\frac{6}{n^2}$ is best possible for a bound uniform over all coefficient vectors $a$ in $\R^n$. To see that, consider the \emph{balanced} coefficient vector $a = (\frac{1}{\sqrt{n}}, \dots, \frac{1}{\sqrt{n}}, -\frac{1}{\sqrt{n}}, \dots, -\frac{1}{\sqrt{n}})$ for even $n=2m$, with half of the coefficients equal to $\frac{1}{\sqrt{n}}$ and half equal to $-\frac{1}{\sqrt{n}}$.  Note that $s_3 = 0$. Then, regrouping,
\[ 
X = \sum_{j=1}^m \frac{1}{\sqrt{m}}Y_j, \qquad Y_j = \frac{X_j-X_{j+m}}{\sqrt{2}} =\frac{\cE_j-\cE_{j+m}}{\sqrt{2}}.
 \]
The random variables $Y_j$ are i.i.d. symmetric exponential, with mean $0$, variance~$1$. Now, by virtue of the aforementioned Bobkov, Chistyakov and G\"otze's asymptotic expansion from \cite{BCG-Fisher}, $J(X) = \frac{\kappa_4(Y_1)^2}{6m^2} + O(n^{-3}) = \frac{4\kappa_4(Y_1)^2}{6n^2} + O(n^{-3})$, and here $\frac{4\kappa_4(Y_1)^2}{6} = \frac{4(2\kappa_4(X_1/\sqrt2))^2}{6} = \frac{\kappa_4(X_1)^2}{6} = 6$.

We finish this section with one more example of coordinate marginals of the uniform distribution on the sphere, interpolating between the uniform ($d=3$) and Gaussian distribution ($d\to \infty$). These random variables have received considerable attention in probabilistic literature, for instance, see \cite{BC, CST, KK, LO-unif}. In particular, the next corollary complements recent results from \cite{CGT, RTW}.

\begin{corollary}\label{cor:unifs}
Let $d \geq 2$, and $\theta$ be a random vector uniform on the unit sphere $S^{d-1}$ in $\R^d$ and let $X_1, X_2, \dots,X_n$ be i.i.d. copies of its variance $1$ normalised coordinate marginal, say $\sqrt{d}\theta_1 = \scal{\sqrt{d}\theta}{e_1}$. Let $a_1, a_2, \dots, a_n$ be real numbers such that $\sum_{j=1}^{n} a_j^2 = 1$. For
\[ 
X = \sum_{j=1}^{n} a_jX_j,
 \]
one has
\begin{align*}
J\left(X \right)
&\geq
\frac{\frac{36}{(d+2)^2}s_4^2}
{
6-\frac{54}{d+2}s_4
+\frac{240}{(d+2)(d+4)}s_6
-\frac{36}{(d+2)^2}s_4^2
}
\\
&\geq
\frac{6}{(d+2)^2n^2}
\left(
1-\frac{54}{6(d+2)n}
+\frac{40}{(d+2)(d+4)n^{3/2}}
-\frac{6}{(d+2)^2n^2}
\right)^{-1}
\\
&=
\frac{6}{(d+2)^2n^2}
+\frac{54}{(d+2)^3n^3}
+O_d(n^{-7/2}).
\end{align*}
\end{corollary}
\begin{proof}
For the cumulants, we have $\kappa_3 = 0$ (by symmetry), $\kappa_4 = -\frac{6}{d+2}$ and $\kappa_6 = \frac{240}{(d+2)(d+4)}$. The first inequality is then a direct application of \eqref{eq:Hermite-Fisher-H3}. Because of the negative coefficient $\kappa_4$, \eqref{eq:J-asympt} is not readily applicable to conclude the second inequality, but we redo the analysis: we first apply $s_6 \leq s_4^{3/2}$ and then verify by a direct standard computation that the resulting expression is increasing as a function of $s_4$ in the feasible interval $s_4 \in (0,1]$, to conclude with the bound $s_4 \geq \frac1n$.
\end{proof}

\begin{remark}
The Hermite--Fisher bounds above can be integrated along the
Ornstein--Uhlenbeck flow to yield lower bounds for the Gaussian entropy
deficit. Let \(X\) be a centred random variable with variance one, finite
sixth moment and finite entropy, and let \(Z\) be a standard Gaussian
random variable independent of \(X\). Setting
\[
\widetilde X_t
=
e^{-t}X+\sqrt{1-e^{-2t}}\,Z,
\]
the integrated de Bruijn identity gives
\[
h(Z)-h(X)
=
\int_0^\infty J(\widetilde X_t)\,dt.
\]
Moreover, for every \(r\geq 3\),
\[
\kappa_r(\widetilde X_t)
=
e^{-rt}\kappa_r(X).
\]
Consequently, integrating \eqref{eq:Hermite-Fisher-H2} and
\eqref{eq:Hermite-Fisher-H3}, respectively, and discarding the new
nonpositive correction terms, yields
\[
h(Z)-h(X)
\geq
\frac{\kappa_3(X)^2}
     {6\bigl(2+|\kappa_4(X)|\bigr)}
\]
and
\[
h(Z)-h(X)
\geq
\frac{\kappa_4(X)^2}
     {8\bigl(
        6+9\kappa_3(X)^2
        +9|\kappa_4(X)|
        +|\kappa_6(X)|
       \bigr)}.
\]
In particular, if \(X_1,X_2,\ldots, X_n \) are i.i.d. centred random variables
with variance one, finite sixth moment and finite entropy and
\[
S_n=\frac{1}{\sqrt n}\sum_{j=1}^n X_j,
\]
then
\[
h(Z)-h(S_n)
\geq
\frac{\kappa_3(X_1)^2}
     {6n\bigl(2+|\kappa_4(X_1)|/n\bigr)}
=
\frac{\kappa_3(X_1)^2}{12n}
+O_{X_1}(n^{-2}).
\]
If \(\kappa_3(X_1)=0\), and $\kappa_4(X_1) \neq 0$, the second estimate instead gives
\[
h(Z)-h(S_n)
\geq
\frac{\kappa_4(X_1)^2}
     {8n^2\bigl(
        6+9|\kappa_4(X_1)|/n
        +|\kappa_6(X_1)|/n^2
       \bigr)}
=
\frac{\kappa_4(X_1)^2}{48n^2}
+O_{X_1}(n^{-3}).
\]
Thus, the vanishing of the third cumulant is also reflected at the
entropy level: the fourth cumulant then yields an entropy deficit of order
\(n^{-2}\).
\end{remark}

\section{Min-entropy power inequalities}
\label{sec:min-entropy}

The cornerstone result put forward by Shannon, the entropy-power inequality (EPI) asserts that
\[ 
N\left(\sum_{j=1}^n X_j\right) \geq \sum_{j=1}^n N(X_j),
 \]
for all independent random vectors $X_1, \dots, X_n$ in $\R^d$. Its foundational importance for the entire digital age simply cannot be overstated, having transformed how we quantify information.  In words, the entropy power of a sum of independent random vectors is at least the sum of their individual entropy powers and equality holds for Gaussian random vectors with proportional covariance
matrices. The EPI can thus be viewed as an information-theoretic
counterpart of the Brunn--Minkowski inequality. This naturally raises
the question of whether analogous inequalities hold for other R\'enyi
entropy powers, and in particular for the min-entropy power
\(N_\infty\).

Bobkov and Chistyakov in \cite{BCh} established an extension to  the family of $p$-entropy-powers for $p>1$,
\[ 
N_p\left(\sum_{j=1}^n X_j\right) \geq \frac{1}{e}p^{\frac{1}{p-1}}\sum_{j=1}^n N_p(X_j),
 \]
which in the limit $p \to \infty$ yields the following min-EPI,
\begin{equation}\label{eq:min-EPI}
N_\infty\left(\sum_{j=1}^n X_j\right) \geq \frac{1}{e}\sum_{j=1}^n N_\infty(X_j).
\end{equation}
This in fact remains valid for random vectors without densities. In general, we introduce the maximum functional
\[ 
M(X) = \inf\left\{ \alpha \geq 0, \ \p{X \in A} \leq \alpha|A| \ \text{for all Borel sets $A$ in $\R^d$} \right\}
 \]
and put
\[ 
N_\infty(X) = M(X)^{-2/d}
 \]
with $N_\infty(X) = 0$ when $M(X) = +\infty$. In particular, when $X$ has a bounded density $f$, we have $M(X) = \|f\|_\infty$ (see \cite{MMX}).

In the scalar case $d=1$, in work \cite{BCh-max} exploiting rather surprising connections to Ball's cube slicing combined with Rogozin's symmetrisation principle of \cite{Rog}, they improved the constant $\frac{1}{e}$ to $\frac12$, which is best possible, attained when $n=2$ and $X_1, X_2$ are uniform on measurable sets (of finite measure), symmetric with respect to a point. Building on this approach, Melbourne and Roberto in \cite{MR}, having developed \emph{local} stability of Ball's cube slicing, applied it and subsequently determined all equality cases in the sharp scalar min-entropy power inequality \eqref{eq:min-EPI} with constant 1/2 (which are essentially, modulo natural symmetries, the said uniform distributions). 

For $d > 1$, Madiman, Melbourne and Xu in \cite{MMX} extended Rogozin's principle to the vector-valued setting and, by exploiting the slicing inequalities of Brzezinski from \cite{Brz},  obtained the following sharp min-EPI,
\begin{equation}\label{eq:min-EPI-Rd}
N_\infty\left(\sum_{j=1}^n X_j\right) \geq c_d\sum_{j=1}^n N_\infty(X_j),
\end{equation}
with
\[ 
c_d = \frac{\Gamma(1+d/2)^{2/d}}{d/2+1}.
 \]
The equality is attained asymptotically as $n \to \infty$, with $X_1, \dots, X_n$ i.i.d. uniform on an Euclidean ball. Curiously, in addition to this asymptotic ``extremiser'', dimension $d=2$ still admits a nonasymptotic ``extremiser'', similar to the scalar case. Our results provide global stability in all dimensions, with the equality conditions as a straightforward consequence.

We extend the definition of the constant $c_d$ by putting $c_1 = \frac12$, so that \eqref{eq:min-EPI-Rd} holds in all dimensions $d \geq 1$.

\begin{theorem}\label{thm:min-EPIs}
Let $d \geq 1$ and $n \geq 2$. Let $X_1, X_2, \dots, X_n$ be independent random vectors in $\R^d$, suppose $\sum_{j=1}^{n} N_\infty(X_j) > 0$, and set
\[ 
a_k = \sqrt{\frac{N_\infty(X_k)}{\sum_{j=1}^n N_\infty(X_j)}}, \quad k = 1, \dots, n.
 \]
Assuming $a_1 \geq a_2 \geq \dots \geq a_n \geq 0$, we have
\begin{equation}\label{eq:min-EPI-stab}
N_\infty\left(\sum_{j=1}^n X_j\right) \geq c_d\Big(1+\delta_d(a)\Big)\sum_{j=1}^n N_\infty(X_j),
 \end{equation}
where
\begin{equation}\label{eq:delta}
\delta_d(a) = \begin{cases} 
8\cdot 10^{-5}\left|a - \frac{e_1+e_2}{\sqrt{2}}\right|, & d=1,\\
\frac12\min\left\{10^{-40}\left|a - \frac{e_1+e_2}{\sqrt{2}}\right|, \frac{1}{76}\sum_{j=1}^n a_j^4\right\}, & d=2,\\
\exp\left(\eta_d\sum_{j=1}^n a_j^4\right)-1, & d \geq 3,
\end{cases}
\end{equation}
with a positive constant $\eta_d$ which depends only on $d$.

In particular, in dimensions $d = 1, 2$ equality holds in \eqref{eq:min-EPI-Rd} if and only if there are two indices $k, l$, a measurable set $A$ in $\R^d$ of positive finite measure and a vector $v$ in $\R^d$ such that $X_k$ is uniform on $A$, $X_l$ is uniform on $v - A$ and the $X_i$, $i \neq k, l$ are each almost surely constant. Moreover, in dimensions $d \geq 3$, inequality \eqref{eq:min-EPI-Rd} is always strict.
\end{theorem}
\begin{proof}
Put $\beta = \sum_{j=1}^n N_\infty(X_j)$. Let $U_1, U_2, \dots$ be i.i.d. random vectors, each uniform on the \emph{unit-volume} Euclidean ball $\bar B_2^d$ in $\R^d$. This normalisation is chosen so that $N_\infty(U_j) = 1$. The scaled vectors $Y_j = a_jU_j$ then satisfy
\[ 
N_\infty(Y_j) = a_j^2N_\infty(U_j) = a_j^2 = \frac{N_\infty(X_j)}{\beta} = N_\infty\left(\frac{1}{\sqrt{\beta}}X_j\right).
 \]
The aforementioned Rogozin-Madiman-Melbourne-Xu's symmetrisation principle gives
\begin{equation}\label{eq:N-bound}
\frac{1}{\beta}N_\infty\left(\sum_{j=1}^n X_j\right)  = N_\infty\left(\sum_{j=1}^n \frac{1}{\sqrt{\beta}}X_j\right) \geq N_\infty\left(\sum_{j=1}^n Y_j\right) = \left\|f_{\sum a_jU_j}\right\|_\infty^{-2/d}.
 \end{equation}
Here, $f_{\sum a_jU_j}$ denotes the probability density of the sum $\sum a_jU_j$ which is even and log-concave; consequently, it attains
its maximum at the origin, i.e. $\left\|f_{\sum a_jU_j}\right\|_\infty = f_{\sum a_jU_j}(0)$. Note also that $\sum a_j^2 = 1$. That value of the density at the origin then has a natural interpretation as the volume of section of the unit-volume ball in the $\ell_\infty(\ell_2)$ space in $(\R^d)^n$ by a ``block'' subspace 
\[ 
H_a = \left\{x = (x_1, \dots, x_n) \in \R^d \times \dots \times \R^d, \ \sum_{j=1}^n a_jx_j = 0\right\}
 \]  
(of codimension $d$),
\[ 
f_{\sum a_jU_j}(0) = \vol_{nd-d}\Big(\bar B_2^d \times \dots \times \bar B_2^d \cap H_a\Big).
 \]
 We refer e.g. to \cite{Brz} for details. In particular, in dimensions $d=1, 2$, this is the volume of section of the cube and polydisc, respectively (see \cite{Ball-cube, OP}). The results of \cite{ENT,GTW, Brz} assert that, respectively,
 \[ 
\vol_{nd-d}\Big((\bar B_2^d \times \dots \times \bar B_2^d) \cap H_a\Big) \leq \begin{cases}
\sqrt{2} - 6\cdot 10^{-5}\left|a - \frac{e_1+e_2}{\sqrt{2}}\right|, & d = 1, \\
2 - \min\left\{10^{-40}\left|a - \frac{e_1+e_2}{\sqrt{2}}\right|, \frac{1}{76}\sum_{j=1}^n a_j^4\right\}, & d = 2, \\
c_d^{-\frac{d}{2}}, & d \geq 3.
\end{cases}
 \]
Our new ingredient is the following strengthening for $d \geq 3$, naturally complementing these works.
\begin{lemma}\label{lm:Brz-deficit}
For $d \geq 3$ there is a positive constant $\kappa_d$ such that for all unit vectors $a$ in $\R^n$, we have
\[ 
\vol_{nd-d}\Big((\bar B_2^d \times \dots \times \bar B_2^d) \cap H_a\Big) \leq c_d^{-\frac{d}{2}}\exp\left(-\kappa_d\sum_{j=1}^n a_j^4\right).
 \]
\end{lemma}
We postpone its proof.
It remains to plug in these volume estimates back into \eqref{eq:N-bound}, combined with the cosmetic elementary inequalities $(\sqrt{2}-t)^{-2} \geq \frac12(1+\sqrt2 t)$ and $(2-t)^{-1} \geq \frac12(1+t/2)$, $0 < t <1$, when $d=1, 2$, respectively. For $d \geq 3$, Lemma \ref{lm:Brz-deficit} yields $\eta_d = \frac{2}{d}\kappa_d$ and $\sum a_j^4 > 0$ justifies the strictness.

The equality conditions in dimensions $d=1, 2$ are justified by noting that then necessarily $\delta_d(a) = 0$, implying $a_1 = a_2 = \frac{1}{\sqrt2}$ and $a_j = 0$, $3 \leq j \leq n$. The rest of the argument proceeds as in \cite{MR} (done for $d=1$, which can be repeated verbatim for $d = 2$, and for completeness, we include that in the appendix).
\end{proof}

%We conjecture that in dimensions $d \geq 3$, inequality \eqref{eq:min-EPI-Rd} only admits an asymptotic ``extremiser'', and this naturally motivates a conjecture that Brzezinski's result admits a global stability result with the deficit term $\delta(a) = \sum_{j=1}^{n} a_j^4$, which will be explored elsewhere.

\begin{proof}[Proof of Lemma \ref{lm:Brz-deficit}]
Put
\[
C_d=c_d^{-d/2}
=\frac{(1+d/2)^{d/2}}{\Gamma(1+d/2)}
\]
and let $U$ be uniformly distributed on $\bar B_2^d$. We first establish a
quantitative form of the Bessel integral estimate used in \cite{Brz}. Writing
\[
\varphi_d(t)=\E e^{i\langle t,U\rangle},
\qquad
\mathcal I_d(p)
=(2\pi)^{-d}\int_{\R^d}|\varphi_d(t)|^p\dd t,
\]
we claim that there is $\varepsilon_d>0$ such that
\begin{equation}\label{eq:Bessel-deficit}
\mathcal I_d(p)
\leq
C_dp^{-d/2}\left(1-\frac{\varepsilon_d}{p}\right),
\qquad p\geq2.
\end{equation}

Let $\rho_d$ be the radius of $\bar B_2^d$. We define the relevant Gaussian function
\[
\gamma_d( t)
=
\exp\left(-\frac{\sigma_d^2| t|^2}{2}\right), \qquad \sigma_d=\frac{\rho_d}{\sqrt{d+2}}.
\]
It is chosen such that
\[
(2\pi)^{-d}\int_{\R^d}\gamma_d( t)^p\dd t
=
C_dp^{-d/2}.
\]
The argument in \cite{Brz}, based on the Nazarov--Podkorytov method
\cite{NP}, gives the following one-crossing property for the distribution
functions of $|\varphi_d|$ and $\gamma_d$. If
\[
F(s)=\vol_d\{ t\in\R^d:|\varphi_d( t)|>s\},
\qquad
G(s)=\vol_d\{ t\in\R^d:\gamma_d( t)>s\},
\]
then, for some $s_0\in(0,1)$, the difference $D=F-G$ is nonnegative on
$(0,s_0)$ and nonpositive on $(s_0,1)$.

Since $\bar B_2^d$ has volume one, Plancherel's identity gives
\[
(2\pi)^{-d}\int_{\R^d}|\varphi_d( t)|^2\dd t=1,
\]
whereas
\[
(2\pi)^{-d}\int_{\R^d}\gamma_d( t)^2\dd t
=
2^{-d/2}C_d>1,
\qquad d\geq3.
\]
The last inequality follows by putting
$q_d=C_d/2^{d/2}$ and observing that
\[
q_3=\frac{5^{3/2}}{6\sqrt{\pi}}>1,
\qquad
q_4=\frac98>1, \qquad
\frac{q_{d+2}}{q_d}
=
\frac12
\left(1+\frac{2}{d+2}\right)^{(d+2)/2}>1.
\]
Consequently,
\[
\int_0^1sD(s)\dd s<0.
\]
For $p>2$, the layer-cake formula and the one-crossing property yield
\[
\begin{split}
\int_0^1s^{p-1}D(s)\dd s
={}&
\int_0^1s\bigl(s^{p-2}-s_0^{p-2}\bigr)D(s)\dd s\\
&\qquad
+s_0^{p-2}\int_0^1sD(s)\dd s<0.
\end{split}
\]
Indeed, the first term is nonpositive because its two nonzero factors have
opposite signs on either side of $s_0$, and the second term is strictly
negative. It follows that
\begin{equation}\label{eq:Bessel-strict}
\mathcal I_d(p)<C_dp^{-d/2},
\qquad p\geq2.
\end{equation}

We next quantify the strictness as $p\to\infty$. We will crucially rely on results from \cite{KOS, Paris}, thus we rewrite the relevant functions in their notation. With $J_\nu$, as usual denoting the Bessel function of the first kind, we have
\[
\varphi_d( t)
=
\frac{\Gamma(1+\nu)J_\nu(\rho_d| t|)}{(\rho_d| t|/2)^\nu}, \qquad \nu = \frac{d}{2}.
\]
Passing to polar coordinates gives
\[
\mathcal I_d(p)
=
\frac{2^{1-2\nu}}{\Gamma(\nu)\Gamma(1+\nu)}
L(\nu;p),
\]
where, in the notation of \cite{Paris},
\[
L(\nu;p)
=
\int_0^\infty
\left|
\frac{\Gamma(1+\nu)J_\nu(r)}{(r/2)^\nu}
\right|^p
r^{2\nu-1}\dd r.
\]
We emphasise that the exponent $p$ is explicitly allowed to be nonintegral in \cite{Paris}. His Laplace-type expansion therefore holds as $p\to\infty$ through arbitrary real values and gives
\[
L(\nu;p)
=
2^{2\nu-1}(1+\nu)^\nu\Gamma(\nu)p^{-\nu}
\left(
1-\frac{\nu(\nu+1)}{2(\nu+2)}\frac1p
+O_\nu(p^{-2})
\right).
\]
Consequently,
\[
\mathcal I_d(p)
=
C_dp^{-d/2}
\left(
1-\frac{\beta_d}{p}+O_d(p^{-2})
\right),
\qquad
\beta_d=\frac{d(d+2)}{4(d+4)}>0.
\]

Define
\[
R_d(p)=\frac{p^{d/2}}{C_d}\mathcal I_d(p).
\]
By \eqref{eq:Bessel-strict}, $R_d(p)<1$ for every $p\geq2$, while the
preceding expansion gives
\[
\lim_{p\to\infty}p\bigl(1-R_d(p)\bigr)=\beta_d.
\]
Moreover, $R_d$ is continuous on $[2,\infty)$ by dominated convergence,
since $|\varphi_d|\leq1$ and $|\varphi_d|^2$ is integrable. It follows,
first on a sufficiently large tail and then on the remaining compact
interval, that
\[
\varepsilon_d =\min\left\{1, \inf_{p\geq2}p\bigl(1-R_d(p)\bigr)
\right\}>0.
\]
This proves \eqref{eq:Bessel-deficit}.

We now turn to the section estimate. Let $U_1,\ldots,U_n$ be independent
copies of $U$, let $f_a$ be the density of
\[
S_a=\sum_{j=1}^na_jU_j,
\]
The probabilistic formula for sections gives
\[
f_a(0)
=
\vol_{nd-d}\Bigl((\bar B_2^d)^n\cap H_a\Bigr).
\]
We define
\[
\kappa_d
=
\min\left\{
\varepsilon_d,
\log\frac{C_d}{2^{d/2}}
\right\}>0.
\]

Suppose first that $|a_k|\geq1/\sqrt2$ for some $k$. Convolution with a
probability measure does not increase the $L_\infty$-norm, and hence
\[
f_a(0)
\leq
\|f_a\|_\infty
\leq
\|f_{a_kU_k}\|_\infty
=
|a_k|^{-d}
\leq
2^{d/2}
\leq
C_de^{-\kappa_d\Delta(a)},
\]
where we have put
\[
\Delta(a)=\sum_{j=1}^na_j^4.
\]
Note that since $a$ is of unit length, $\Delta(a) \leq 1$, which, together with the definition of $\kappa_d$ justifies the last inequality above.

It remains to consider the case
\[
|a_j|<\frac1{\sqrt2}
\]
for all $j$ (without loss of generality, we can assume that the $a_j$ are nonzero, otherwise, we simply skip those indices). Put $p_j=a_j^{-2}>2$, so that $\sum \frac{1}{p_j} = 1$.
Fourier inversion and H\"older's inequality thus give
\[
f_a(0)
\leq
\prod 
\left(
(2\pi)^{-d}
\int_{\R^d}
|\varphi_d(a_j t)|^{p_j}\dd t
\right)^{1/p_j}.
\]
By a change of variables and \eqref{eq:Bessel-deficit},
\[
\begin{split}
(2\pi)^{-d}
\int_{\R^d}
|\varphi_d(a_j t)|^{p_j}\dd t =
|a_j|^{-d}\mathcal I_d(p_j) \leq
|a_j|^{-d}C_dp_j^{-d/2}
\left(1-\frac{\varepsilon_d}{p_j}\right) =
C_d(1-\varepsilon_da_j^2).
\end{split}
\]
Therefore, 
\[
\begin{split}
f_a(0) \leq
C_d\prod_j(1-\varepsilon_da_j^2)^{a_j^2} \leq
C_d\exp\left(
-\varepsilon_d\sum_{j=1}^na_j^4
\right) \leq
C_d\exp\left(
-\kappa_d\sum_{j=1}^na_j^4
\right).
\end{split}
\]
This finishes the proof.
\end{proof}

\appendix

\section*{Appendix: Equality conditions in \eqref{eq:min-EPI-Rd}}\label{app2}

Once
\[
a_1=a_2=\frac{1}{\sqrt{2}},
\qquad
a_j=0,\quad j\geq 3,
\]
has been established (as a consequence of $\delta(a) = 0$), the remaining equality analysis works in every dimension. We repeat the Melbourne-Roberto argument from \cite{MR}.

In the normalised setting, let $f,g$ be the densities of the two nondegenerate summands and let
\[
\|f\|_\infty=\|g\|_\infty=m.
\]
Equality gives
\[
\|f*g\|_\infty=m.
\]
Since $f*g$ is continuous and vanishes at infinity, it attains its maximum, say at $v\in\mathbb R^d$, and
\[
m=(f*g)(v)
=\int_{\mathbb R^d}f(x)g(v-x)\,dx
\leq m\int_{\mathbb R^d}f(x)\,dx
=m.
\]
Consequently,
\[
g(v-x)=m
\]
for $f$-almost every $x$. Let
\[
A=\{x\in\mathbb R^d:f(x)>0\},
\qquad
B=\{x\in\mathbb R^d:g(x)=m\}.
\]
Since $f\leq m$ and $\int f=1$, whereas $g=m$ on $B$ and $\int g=1$, we have
\[
|A|\geq\frac{1}{m},
\qquad
|B|\leq\frac{1}{m}.
\]
Moreover, the preceding equality implies
\[
v-A\subset B
\]
modulo null sets. Hence
\[
\frac{1}{m}\leq |A|=|v-A|\leq |B|\leq\frac{1}{m},
\]
so that
\[
|A|=|B|=\frac{1}{m}.
\]
Since $0\leq f\leq m$, $f$ is supported on $A$, and $\int f=1=m|A|$, it follows that
\[
f=m\mathbf 1_A
\qquad\text{a.e.}
\]
Similarly,
\[
g=m\mathbf 1_B
=m\mathbf 1_{v-A}
\qquad\text{a.e.}
\]
Thus the two nondegenerate summands are uniform on $A$ and $v-A$, respectively.

\subsection*{Acknowledgements.}
Part of this work was developed during TT's visit to the Chinese University of Hong Kong, gratefully appreciating Chandra Nair's hospitality and the welcoming environment provided by the Department of Information Engineering.

%\subsection*{AI statement.}
%The authors used generative AI tools solely for language editing and proofreading.


\begin{thebibliography}{9}


\bibitem{AG}
Ambrus, G., Gárgyán, B.,
Minimal central slices of the regular simplex.
Preprint (2026), arXiv:2609.12714.


\bibitem{ABBN1}
Artstein, S., Ball, K., Barthe, F., Naor, A.,
Solution of Shannon's problem on the monotonicity of entropy.
J. Amer. Math. Soc. 17 (2004), no. 4, 975--982.

\bibitem{ABBN2}
Artstein, S., Ball, K., Barthe, F., Naor, A.,
On the rate of convergence in the entropic central limit theorem. (English summary)
Probab. Theory Related Fields 129 (2004), no. 3, 381--390.


\bibitem{Ball-cube}
Ball, K.,
Cube slicing in $\R^n$. 
Proc. Amer. Math. Soc. 97 (1986), no. 3, 465--473.



\bibitem{BC}
Baernstein, A., II, Culverhouse, R., Majorization of sequences, sharp vector Khinchin inequalities, and bisubharmonic functions. Studia Math. 152 (2002), no. 3, 231--248.

\bibitem{BCh} 
Bobkov, S. G., Chistyakov, G. P.,
Entropy power inequality for the R\'enyi entropy,
IEEE Trans. Inform. Theory, 61 (2015), no. 2, 708--714. 


\bibitem{BCG-Fisher}
Bobkov, S. G., Chistyakov, G. P., Götze, F.,
Fisher information and the central limit theorem. Probab. Theory Related Fields 159 (2014), no. 1--2, 1--59.


\bibitem{BCG-book}
Bobkov, S. G., Chistyakov, G. P., Götze, F.,
Concentration and Gaussian approximation for randomized sums.
Probability Theory and Stochastic Modelling, 104. Springer, Cham, 2023.


\bibitem{BCh-max}
Bobkov, S. G., Chistyakov, G. P.,
Bounds on the maximum of the density for sums of independent random
variables.
J. Math. Sci. (N.Y.) 199 (2014), no.~2, 100--106.


\bibitem{BLM}
Boucheron, S., Lugosi, G., Massart, P., Concentration inequalities. A nonasymptotic theory of independence. With a foreword by Michel Ledoux. Oxford University Press, Oxford, 2013.


\bibitem{BP1}
Brazitikos, S., Pandis, C.,
Sharp inequalities for symmetric polynomials, Hunter's conjecture, and moments of exponential random variables.
Preprint (2025), arXiv:2512.12254.


\bibitem{BP2}
Brazitikos, S., Pandis, C.,
Restricted hyperplane sections of the cross-polytope and the simplex.
Preprint (2026), arXiv:2606.07163.


\bibitem{BTT}
Brazitikos, S., Tang, C., Tkocz, T.,
Moments of sums of exponentials, beyond CHS.
Preprint (2026), arXiv:2602.03058.



\bibitem{Brz}
Brzezinski, P.,
Volume estimates for sections of certain convex bodies.
Math. Nachr. 286 (2013), no. 17--18, 1726--1743.

\bibitem{CGT}
Chasapis, G., Gurushanka, K., Tkocz, T., Sharp bounds on p-norms for sums of independent uniform random variables, $0<p<1$. J. Anal. Math. 149 (2023), no. 2, 529--553. 

\bibitem{CST}
Chasapis, G., Singh, S., Tkocz, T.,
Haagerup's phase transition at polydisc slicing. Anal. PDE 17 (2024), no. 7, 2509--2539. 

\bibitem{CT}
Cover, T. M., Thomas, J. A.,
Elements of information theory,
second ed., Wiley-Interscience, Hoboken, NJ, 2006.


\bibitem{EG}
Eskenazis, A.,  Gavalakis, L.,
On the entropy and information of Gaussian mixtures.
Mathematika 70 (2024), no. 2, Paper No. e12246, 19 pp.


\bibitem{ENT-GM}
Eskenazis, A., Nayar, P., Tkocz, T.,
Gaussian mixtures: entropy and geometric inequalities. Ann. Probab. 46 (2018), no. 5, 2908--2945.


\bibitem{ENT}
Eskenazis, A., Nayar, P., Tkocz, T.,
Resilience of cube slicing in $\ell_p$.
\emph{Duke Math. J.} 173 (2024), no.~17, 3377--3412.

\bibitem{GTW}
Glover, N., Tkocz, T., Wyczesany, K.,
Stability of polydisc slicing.
Mathematika 69 (2023), no.~4, 1165--1182.



\bibitem{Jarr}
Jarrett, R. G.,
Bounds and expansions for Fisher information when the moments are known.
Biometrika 71 (1984), no. 1, 101--113.


\bibitem{KOS}
Kerman, R., Oľhava, R., Spektor, S.,
An asymptotically sharp form of Ball's integral inequality. 
Proc. Amer. Math. Soc. 143 (2015), no. 9, 3839--3846.



\bibitem{KK}
K\"onig, H., Kwapie\'n, S.,
Best Khintchine type inequalities for sums of independent, rotationally invariant random vectors.
Positivity 5 (2001), no. 2, 115--152.



\bibitem{LO-unif}
Lata\l a, R., Oleszkiewicz, K.,
A note on sums of independent uniformly distributed random variables.
Colloq. Math. 68 (1995), no. 2, 197--206.


\bibitem{Led}
Ledoux, M.,
The concentration of measure phenomenon.
Mathematical Surveys and Monographs, 89. American Mathematical Society, Providence, RI, 2001.




\bibitem{MMX}
Madiman, M., Melbourne, J., Xu, P., 
Rogozin's convolution inequality for locally compact groups.
Preprint (2017): arXiv:1705.00642.



\bibitem{MR}
Melbourne, J., Roberto, C.,
Quantitative form of Ball's cube slicing in $\R^n$ and equality cases in the min-entropy power inequality.
Proc. Amer. Math. Soc. 150 (2022), no. 8, 3595--3611.



\bibitem{MRTT}
Melbourne, J., Roysdon, M., Tang, C., Tkocz, T.,
From simplex slicing to sharp reverse Hölder inequalities.
J. Lond. Math. Soc. (2) 113 (2026), no. 3, Paper No. e70461, 23 pp.



\bibitem{MTTT}
Myroshnychenko, S., Tang, C., Tatarko, K., Tkocz, T.,
Stability of simplex slicing.
Discrete Comput. Geom. 76 (2026), no. 1, 491–507.

\bibitem{NP}
Nazarov, F.L., Podkorytov, A.N., Ball, Haagerup, and distribution functions, in Complex Analysis, Operators, and Related Topics, Oper. Theory Adv. Appl. 113 (2000), 247–267.

\bibitem{OP}
Oleszkiewicz, K., Pe\l czy\'nski, A.,
Polydisc slicing in $\mathbb{C}^n$.
Studia Math. 142 (2000), no. 3, 281--294.


\bibitem{Paris}
Paris, R. B.,
Asymptotics of some generalised sine-integrals.
Preprint (2021), arXiv:2011.05156.


\bibitem{RTW}
Rapaport, M., Tkocz, T., Wu, I.,
Negative Moments of Steinhaus Sums.
Preprint (2025), arXiv:2512.09077.


\bibitem{Renyi1961}
R{\'e}nyi, A.,
On measures of entropy and information.
In \emph{Proceedings of the Fourth Berkeley Symposium on Mathematical
Statistics and Probability}, Vol.~1, University of California Press,
Berkeley, 1961, 547--561.




\bibitem{Rog}
Rogozin, B. A., An estimate for the maximum of the convolution for bounded densities, Teor. Veroyatn. Primen., 32 (1987), 53--61.


\bibitem{Stam}
Stam, A. J.,
Some inequalities satisfied by the quantities of information of Fisher and Shannon.
Information and Control 2 (1959), 101--112.



\bibitem{Tang}
Tang, C., 
Simplex slicing: An asymptotically-sharp lower bound,
Adv. Math. 451 (2024), 109784.


\bibitem{Webb}
Webb, S., 
Central slices of the regular simplex. 
Geom. Dedicata 61 (1996), no. 1, 19--28.


\end{thebibliography}
\end{document}